\documentclass[12pt,draftclsnofoot,onecolumn,peerreview]{IEEEtran}
\usepackage{cite}
\usepackage[pdftex]{graphicx}
\usepackage[pdftex]{color}
\graphicspath{/}
\usepackage[ruled,vlined,linesnumbered]{algorithm2e}
\usepackage[cmex10]{amsmath}
\usepackage{algpseudocode}
\usepackage{array}
\usepackage{url}
\usepackage{bm} 
\usepackage{amssymb}

\usepackage{amsthm}
\newtheorem{thm}{Theorem}[section]

\newtheorem{lem}[thm]{Lemma}
\newtheorem{prop}[thm]{Proposition}

\theoremstyle{remark}

\theoremstyle{definition}

\newcommand{\C}[1]{C\left( #1 \right)}

\begin{document}
\title{Zero-Outage Cellular Downlink with Fixed-Rate D2D Underlay}

\author{\IEEEauthorblockN{Nuno K. Pratas and Petar Popovski}
\IEEEauthorblockA{\\Department of Electronic Systems, Aalborg University, Denmark \\
Email: \{nup,petarp\}@es.aau.dk}}

\maketitle


\begin{abstract}
Two of the emerging trends in wireless cellular systems are Device-to-Device (D2D) and Machine-to-Machine (M2M) communications. D2D enables efficient reuse of the licensed spectrum to support localized transmissions, while M2M connections are often characterized by fixed and low transmission rates. D2D connections can be instrumental in localized aggregation of uplink M2M traffic to a more capable cellular device, before being finally delivered to the Base Station (BS). In this paper we show 
that a fixed M2M rate is an enabler of efficient Machine-Type D2D underlay operation taking place simultaneously with another  \emph{downlink} cellular transmission.
In the considered scenario, a BS $B$ transmits to a user $U$, while there are $N_M$ Machine-Type Devices (MTDs) attached to $U$, all sending simultaneously to $U$ and each using the same rate $R_M$. While assuming that $B$ knows the channel $B-U$, but not the interfering channels from the MTDs to $U$, we prove that there is a positive downlink rate that can always be decoded by $U$, leading to zero-outage of the downlink signal. This is a rather surprising consequence of the features of the multiple access channel and the fixed rate $R_M$. We also consider the case of a simpler, single-user decoder at $U$ with successive interference cancellation. However, with single-user decoder, a positive zero-outage rate exists only when $N_M=1$ and is zero when $N_M>1$. This implies that joint decoding is instrumental in enabling fixed-rate underlay operation.


\end{abstract}



\section{Introduction}
\label{sec:Introduction}


This work is motivated by two technology trends in wireless cellular networks~\cite{6736746}: direct \emph{Device-to-Device (D2D)} communications~\cite{6231164,6163598,DBLP:journals/corr/AsadiWM13} and \emph{Machine-Type Communications (MTC)} or \emph{Machine-to-Machine (M2M)} communications~\cite{boswarthick2012m2m}. D2D communication refers to the direct links between the wireless devices using the same spectrum and air interface as in cellular communications. M2M involves a large number of low-rate low-power \emph{Machine-Type Devices (MTDs)}, attached to the cellular network, enabling various applications, such as asset/health monitoring, smart grid communications, large-scale environmental sensing, etc.

\subsection{State of the Art} 
\label{sec:state_of_the_art}

Recently, several potential synergies between D2D and M2M have been investigated~\cite{Pratas2014,2013arXiv1305.6783P,6757899}, aiming to increase the spectral reuse, range and coverage. In~\cite{Chen2013} these synergies where explored in the context of group-based operations, i.e. of a D2D-enabled cellular device that acts as a clusterhead for a group of MTDs.
In the uplink the clusterhead aggregates and forwards the gathered requests, data packets and status information from MTC devices to the connected 3GPP Base Station (BS). In the downlink, the cluster head relays management messages and data packets from the BS to the MTDs in the group.
The use of D2D communications reduces the signaling congestion on the air interface and the network management load, which goes in-line with the spirit of the group-based management defined by 3GPP~\cite{3GPPTS22.368}.
In~\cite{5506183} it was considered how direct communication within a cluster of devices can improve the performance of a conventional cellular system, by leveraging the selection between direct D2D and infrastructure relaying.
In~\cite{6497010}, the impact of the MTDs on the cellular network is mitigated by randomly deploying data collectors to gather the traffic from the MTDs.
In~\cite{6786066}, the authors use a queueing model and coalition game formulation to analyze a scenario where the MTDs are able to transmit to a macro or small-cell BS, or perform relay transmission. One of the observation was that the overall throughput of the MTDs is higher if they have a low duty cycle.

Standardization considers D2D under the name Proximity Services (ProSe)~\cite{3GPPTR22.803}, where among several other use cases, it highlights the use case in which a cellular device improves the coverage by acting as a relay on behalf of one or more other cellular devices that are outside the network range. From architecture viewpoint, D2D relaying in ProSe is referred to as 
\emph{range extension} and is specific to public safety use cases~\cite{3GPPTR23.703}.
The report~\cite{3GPPTR36.843}, besides providing design aspects for the physical and the upper layers, it also states that one of the requirements is the support of a large number of concurrently participating ProSe-enabled users.


\subsection{Our Contribution} 
\label{sec:our_contribution}

\begin{figure}
	\centering
		\includegraphics[width=\linewidth]{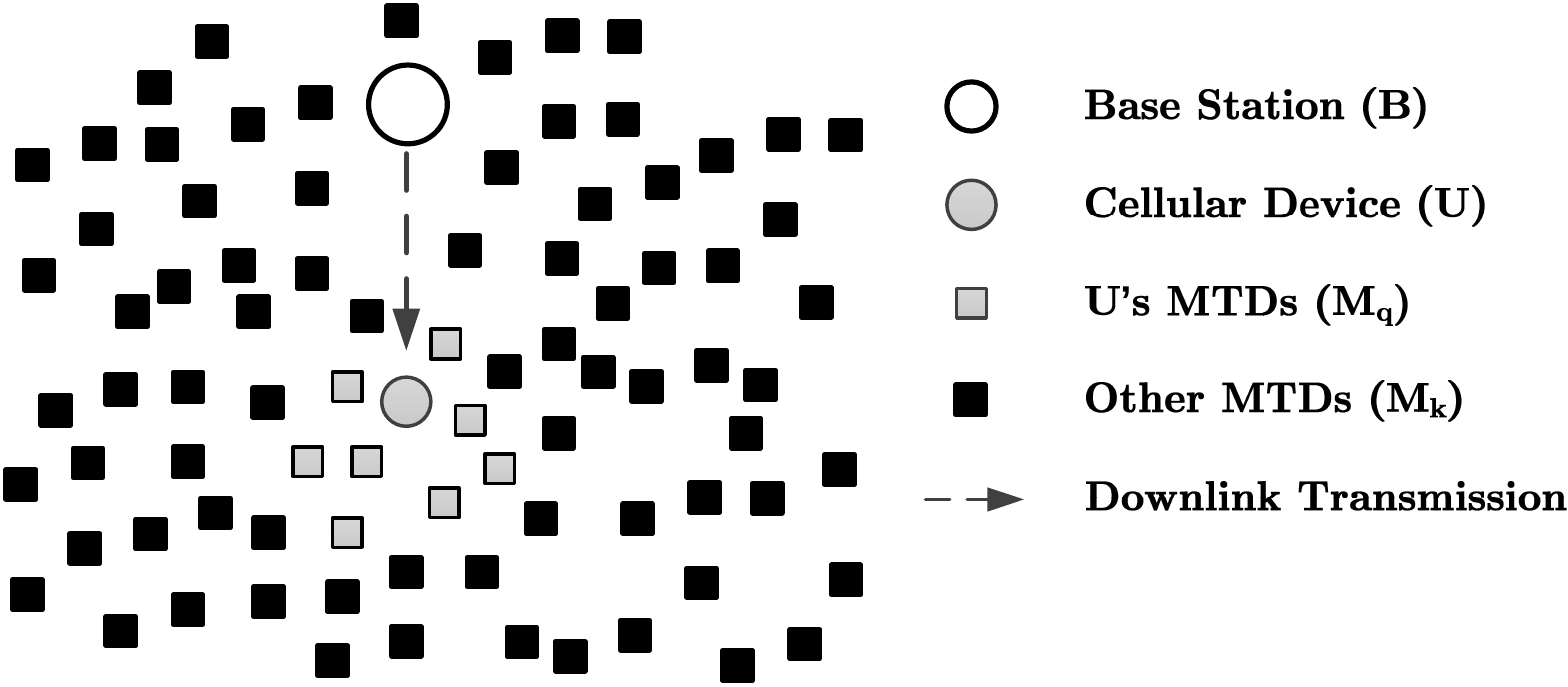}
	\caption{The cellular user $U$, receives simultaneously, in the same time and frequency resource, the transmissions from the Base Station $B$ and from the Machine-Type Devices $M_q$ associated with $U$. The black devices are MTDs whose signals are not decoded by $U$ and create background non-cancelable interference to $U$.}
	\label{fig:Scenario}
\end{figure}
A motivating scenario for this work is the one in which one or multiple MTDs are attached to the cellular network through a D2D connection to a more capable device. For example, one can think health sensors as MTDs that are attached to the body and should supply data to the cellular network about the monitored user. Each sensor is capable to communicate directly with the cellular infrastructure; however, if the person has a smartphone at disposal, the sensor uses a D2D link to transmit the data to the smartphone. Speaking in more general terms, in our scenario, the D2D link is used by an M2M device to send uplink traffic to another, more capable device. The latter may aggregate the M2M traffic from one or multiple M2M devices and then send it in the uplink to the BS. The architecture in which the M2M traffic is locally aggregated to a more capable device rather than being sent directly over a wide-area connection has at least two advantages. First, the uplink power can be lowered. Second, there are less M2M devices that contend directly to get uplink resources at the BS. The latter is important if we consider a massive number of M2M devices in a single cell, such that they cannot be served if activated approximately at the same time. However, using D2D links and traffic aggregation, the resource sharing occurs locally, among a relatively lower number of M2M devices.

Fig.~\ref{fig:Scenario} depicts out the network-assisted D2D solution scenario~\cite{6163598}, which enables concurrent use of the same communication resource by MTDs and normal cellular devices.
The M2M connection is assumed to have a fixed, low rate $R_M$. We also consider downlink traffic that is sent simultaneously with the D2D transmission. The downlink traffic is not necessarily a M2M traffic and therefore the rate is adaptable, aiming to maximize the downlink throughput according with the channel conditions. 
Our scenario combines the following four aspects in a D2D setting: (1) the underlay of \emph{downlink} cellular communications~\cite{Pratas2014,6214244,5910123,5199353,PekkaJANIS2009}; (2) the use of interference cancellation techniques~\cite{Pratas2014,6214244}; (3) the use of D2D to support a fixed and \emph{low-rate M2M} connection~\cite{Pratas2014,6175090,5910120} and (4) considers the case of partially available Channel State Information at the Transmitter (CSIT)~\cite{Pratas2014,5450264,5073734}. 

We show that, despite the activation of the MTDs transmission and the fact that the interfering channel $M_q-U$ is not known, it is still possible for the BS $B$ to select a positive downlink rate for the cellular link $B-U$ that experiences zero outage. 
In other words, the downlink signal is always decoded by $U$, regardless of the strength of the interfering channel from $M_q$ to $U$. The most interesting aspect of it is that $B$ only needs to know the instantaneous channel $B-U$, but it does not need to know anything about the channel $M_q-U$, not even the channel statistics. We show that $B$ needs only to know the MTD rate $R_M$, which is assumed fixed and a priori known by all. Intuitively, this is possible because, for fixed $R_M$, if the link $M_q-U$ is weak, then the interference from $M_q$ is treated as noise, but if it is very strong, then the signal from $M_q$ is decoded by $U$ and thus the interference cancelled. The maximal zero-outage downlink rate occurs in the transition between these two regimes. We first consider the full decoding region of a multiple access channel, referred to as Joint-User Decoding (JD), and prove a closed-form formula of the maximal zero-outage downlink rate when there are $N_M$ MTDs attached to $U$. For the special case $N_M=1$, we also consider a simpler, Single-User Decoding (SD), where $U$ needs to use only single-user decoding. 
This work is a generalization of the scenario first presented in~\cite{Pratas2014,2013arXiv1305.6783P}, where the M2M and D2D coexistence was analyzed in a different network topology and assumptions.


The paper is organized as follows.
In Section~\ref{sec:SystemModel} we present the system model, 
followed by Section~\ref{sec:ZeroOutageRateUpperbound}, where we derive the maximal zero outage downlink rate in both JD and SD.
In Section~\ref{sec:PerformanceAnalysis}, we provide analytical bounds and approximations for the performance metrics using stochastic geometry tools.
In Section~\ref{sec:Numerical Results}, we provide the comparison between the analytical bounds and approximations and the Monte Carlo simulations results.
Finally, we conclude the paper in Section~\ref{sec:Conclusion}.


\section{System Model}
\label{sec:SystemModel}

Fig.~\ref{fig:Scenario} depicts the scenario, focused on the cellular user $U$ that receives downlink traffic from a BS denoted by $B$. 
There are $N_M$ MTDs associated with $U$\footnote{We note that the MTDs transmit their own data and are not acting as relays to $B$'s transmission.}, denoted by $M_i$, $i=1 \ldots N_M$.
Each device transmits to $U$ during the downlink transmission of $B$, such that $U$ observes a multiple access channel (MAC) with $N_M+1$ transmitters:
\begin{equation} \label{y_u}
	y_{U} = h_{B} x_{B} + \sum_{i=1}^{N_M} h_{M_i} x_{M_i} + \tilde{z}
\end{equation}
where $h_{B}$ and $h_{M_i}$ are the complex gains of the channels $B-U$ and $M_i-U$, respectively. $x_{B}$ and $x_{M_i}$ are given respectively by the circular zero-mean Gaussian complex signal transmitted by the $B$, $M$ and $M_i$ nodes, such that the respective variances are $E[|x_{B}|^2] = P_B$ and $E[|x_{M_i}|^2] = P_{M}$, where $P_B$ and $P_M$ represent the constant power levels used by $B$ and $M_i$, respectively.
Finally, $\tilde{z}$ is a complex Gaussian variable that contains the noise as well as the Gaussian-approximated interference from the other MTDs not associated with $U$.
We assume $E[|\tilde{z}|^2]=\tilde{\sigma}^2$ and in Section~\ref{sec:PerformanceAnalysis} we will explicitly address the modeling of the interference.

The Signal-to-Noise Ratio (SNR) $\gamma_{i}$ for the signal transmitted from node $i$, where $i \in \{B, M_1, M_2, \ldots M_{N_M}\}$, to $U$ is defined as
\begin{equation}
	\gamma_{i} = \frac{P_i |h_i|^2}{\tilde{\sigma}^2}
\end{equation}
the channel coefficient $|h_i|^2$ contains the path loss and the fading, and where $P_i=P_M$ for $i \neq B$. All links are assumed to be non-Line-of-Sight and characterized by block Rayleigh fading, such that both the channel fading gains and aggregated interference do not vary during a slot in which a single 
packet is sent. The model of the path loss and fading is presented in Section~\ref{sec:PerformanceAnalysis}.

All transmissions have normalized bandwidth of $1$ Hz, therefore the time duration can be measured in number of symbols.
All transmitters use capacity-achieving Gaussian codebooks, such that if a link has a SNR of $\gamma$ in a given slot, then the maximal achievable rate is 
\begin{equation}
	R = \log_2(1 + \gamma) = C(\gamma)
\end{equation}
Similarly, for a given rate $R_i$ we can compute the minimal required SNR $\Gamma_i$ that the link should have in order for the receiver to decode the signal successfully in absence of other interfering signals. This is given as:
\begin{equation}
	\Gamma_i = C^{-1}(R_i) = 2^{R_i} - 1
\end{equation}
The downlink rate $R_B$ can be adapted from slot to slot, while each MTD uses always a constant rate $R_M$. When there is no danger of causing confusion and for the sake of brevity, we will sometimes refer to $\Gamma_i$ as ``rate''. 
Further, we will denote the SNR of the $M_i$ MTD simply as $\gamma_i$.

We assume that $B$ knows the instantaneous Channel State Information (CSI) of the $B-U$ link, represented by the knowledge of $\gamma_B$. 
The most interesting assumption is that $B$ does not have any knowledge of the channel realization or channel statistics of the link $M_i-U$. $B$ only knows the rate $R_M$ employed by the MTDs, which can be expressed as $\Gamma_M=C^{-1}(R_M)$.
We note that $U$ estimates the channel realizations from all links associated with the transmissions that it decodes, i.e. $U$ has CSI at the Receiver, obtained through e.g. unique preambles of the transmitting MTDs and $B$.

\section{Selection of the Zero-Outage Downlink Rate $R_B$}
\label{sec:ZeroOutageRateUpperbound}

\subsection{Single MTD: $N_M=1$}
\label{sec:SingleMTDNM1}

We start with the case $N_M=1$ and show how to select the maximal downlink rate $R_B$, for given $\gamma_B$, such that the signal sent from the $B$ experiences zero-outage, i.e. $U$ is able to decode it with probability $1$, regardless of the interference caused by the MTDs. The intuition why such a rate exists can be explained as follows. If the interference from $M_1$ is weak, then $\gamma_1$ is low and this interference should be treated as noise. As $\gamma_1$ grows while $R_M$ (i.e. $\Gamma_M$) stays fixed, then the MTD signal becomes decodable, such that it can be removed.  
Hence, there are two operating regimes for the decoder at $U$. We will show that $B$ can always select a positive rate $R_B$ that does not result in outage, although $B$ does not even know the operating regime of the decoder of $U$. 
We denote this decoding setting as Joint User Decoding (JD).

In a model of a standard multiple access channel (MAC) the assumption is that the receiver should decode successfully the signals of all the transmitters. In our case, $U$ is the receiver of two signals with rates $R_B$ and $R_{1}=R_M$, such that the MAC inequalities are:
\begin{align}
	R_M &\leq \C{\gamma_{1}} \nonumber \\ 
	R_B &\leq \C{\gamma_B} \nonumber \\ 
	R_M + R_B &\leq \C{\gamma_{1} + \gamma_B} \label{eq:MACinequalities}
\end{align}
\begin{figure}[tb]
	\begin{center}
		\includegraphics[width=\linewidth]{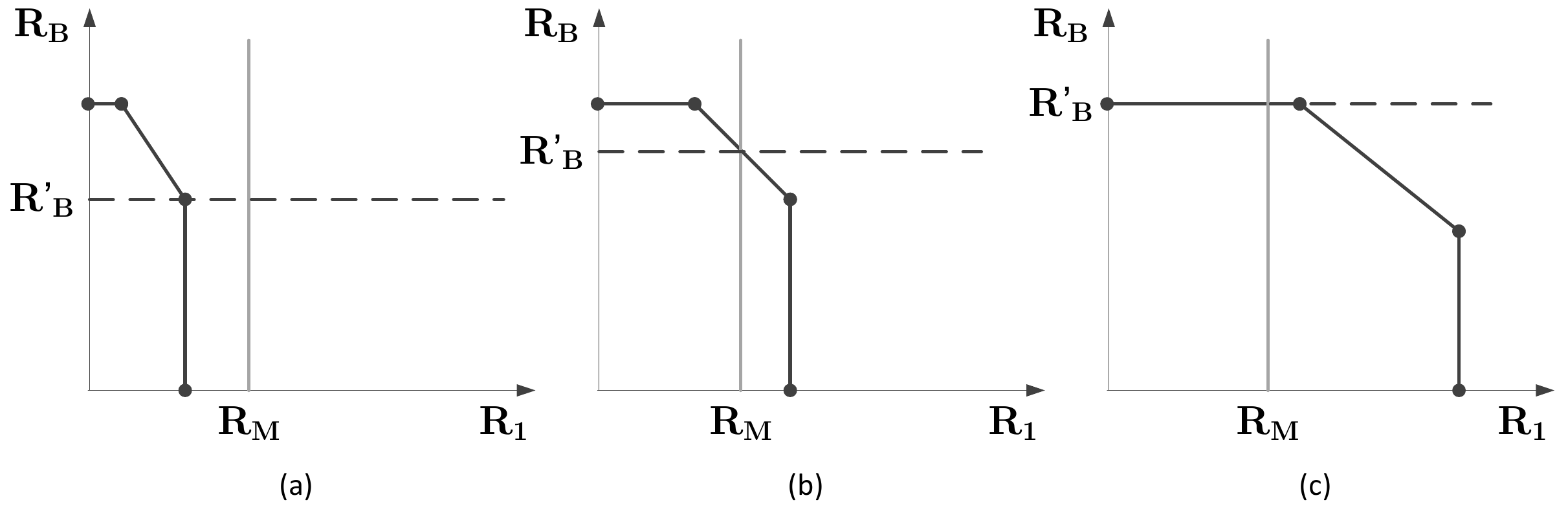}
	\end{center}
	\caption{Illustration of the working regimes of a Joint decoder, where one of the links has fixed rate. The transmission from $B$: (a) is decoded while treating the MTD transmission as noise; (b) is jointly decoded with the transmission from the MTD; (c) is decoded in the presence of noise, after the transmission from the MTD has been decoded and subtracted from the composite signal.}
	\label{fig:JDRegime}
\end{figure}
Our objective is different from a standard MAC treatment, since $R_M$ is fixed and we need to analyze the decodability of $R_B$, regardless of whether $R_M$ is decodable. Fig.~\ref{fig:JDRegime} depicts the three working regimes of the Joint Decoder. If $\gamma_{1}<\Gamma_M$, then the first inequality in (\ref{eq:MACinequalities}) is violated, such that the signal from $M_1$ is treated as noise and the maximal achievable downlink rate is:
\begin{equation}\label{eq:TreatGMasNoise}
	R_B \leq C \left(\frac{\gamma_B}{1+\gamma_{1}} \right)
\end{equation}
On the other hand, when $\gamma_{1} \geq \Gamma_M$ then there are two options: (a) either continue to treat the signal from $M_1$ as noise and use~(\ref{eq:TreatGMasNoise}) to determine the maximal $R_B$ or (b) employ the MAC inequalities~(\ref{eq:MACinequalities}), which leads to the following:
\begin{equation}\label{eq:DecodeGM}
	R_B \leq \min \{ \C{\gamma_B},  \C{\gamma_{1} + \gamma_B} - R_M\}=\min \{ \C{\gamma_B},  \C{\gamma_{1} + \gamma_B} -C(\Gamma_M)\}
\end{equation}

However, note that if for given $\gamma_{1} \geq \Gamma_M$ the expression~(\ref{eq:TreatGMasNoise}) leads to a higher bound on $R_B$ compared to~(\ref{eq:DecodeGM}), then $U$ should treat the signal from $M_1$ as noise (but we will see that this is not the case). With a slight abuse of notation, we can define $R_B (\gamma_{1})$ as the maximal achievable $R_B$ for given $\gamma_{1}$, which can then be written compactly as: 
\begin{equation}\label{eq:RBcompact}
	R_B(\gamma_{1}) = \left\{
  \begin{array}{lr}
    C \left(\frac{\gamma_B}{1+\gamma_{1}} \right) &  \gamma_{1} <  \Gamma_M\\
    \max\{C \left(\frac{\gamma_B}{1+\gamma_{1}} \right), \min \{ \C{\gamma_B},  \C{\gamma_{1} + \gamma_B} -C(\Gamma_M)\}  \}  &  \gamma_{1} \geq \Gamma_M \\
  \end{array}
\right.
\end{equation}
\begin{figure}
	\centering
		\includegraphics[width=0.8\linewidth]{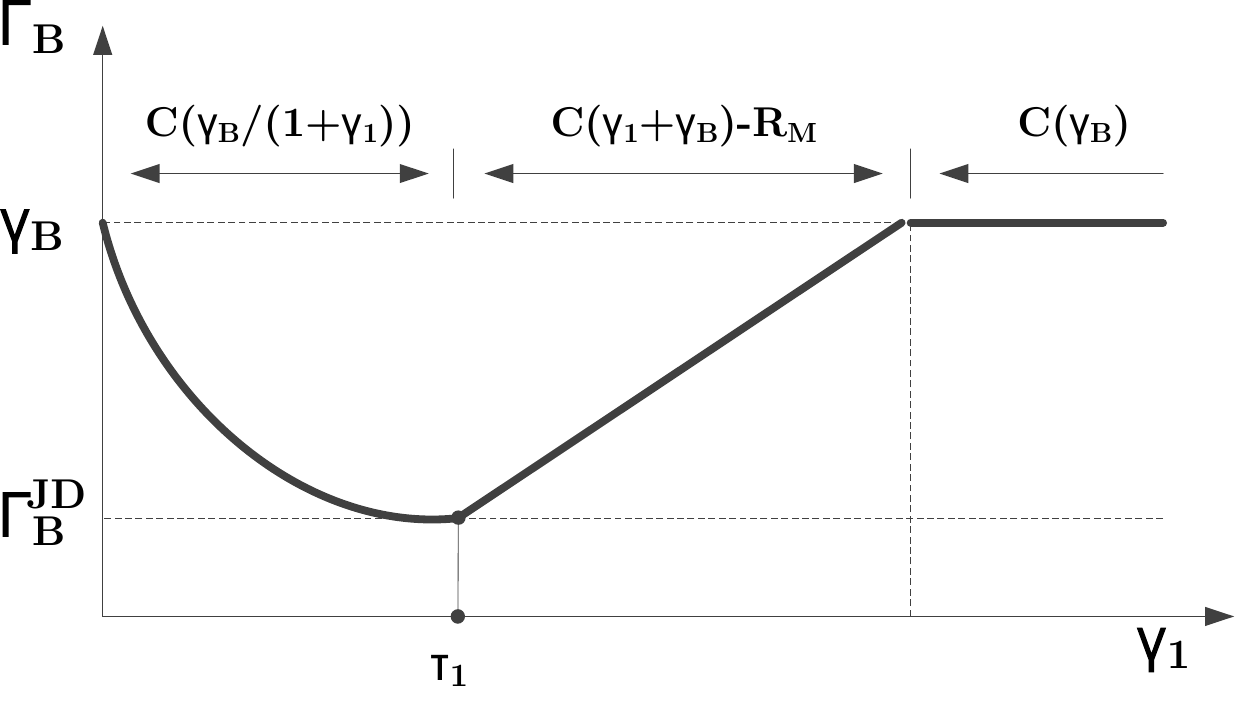}
	\caption{Achievable $\Gamma_B$ when using Joint User Decoding (JD) given known $\gamma_B$ that does not put the $B-U$ link in outage.}
	\label{fig:RateRegionJD_PoB}
\end{figure}
Fig.~\ref{fig:RateRegionJD_PoB} depicts how this function looks like. When $\gamma_{1}=\tau_1<\Gamma_M$, $R_B(\gamma_{1})$ decreases, since a larger $\gamma_{1}$ means a larger noise, corresponding to regime depicted in Figure~\ref{fig:JDRegime}(a). After $\gamma_1$ reaches $\Gamma_M$, $R_B(\gamma_{1})$ starts to increase, corresponding to regime depicted in Figure~\ref{fig:JDRegime}(b). Finally, when the link $M_1-U$ becomes too strong $\gamma_1 \geq \tau_2$, the downlink rate reaches its maximal possible value $R_B(\gamma_{1})=C(\gamma_B)$ where the interference effect from $M_1-U$ vanishes, corresponding to regime depicted in Figure~\ref{fig:JDRegime}(c).

We can now state the following:
\begin{lem}\label{prop:GammaB_JUD}
	Let there be $N_M=1$ MTD sending at rate $R_M$. Let $B$ know the rate $R_M$ and the SNR $\gamma_B$, but not $\gamma_{1}$. Then the maximal downlink transmission rate that is always 	decodable by $U$ is $R_B = C(\Gamma^{JD}_B)$ where
	\begin{equation}\label{eq:GammaBSingleUser}
		\Gamma^{JD}_B = \frac{\gamma_B}{1 + \Gamma_M}
	\end{equation}
and $\Gamma_M=C^{-1}(R_M)$.  
\end{lem} 
\begin{proof}
In order to prove the lemma, we need to show that the $\min_{\gamma_{1} \geq 0} R_B(\gamma_{1})=C(\Gamma^{JD}_B)$. This is because if $C(\Gamma_B) \leq R_B(\gamma_{1})$, then $U$ can always decode the signal from $B$. 

Let $0\leq \gamma_{1} < \Gamma_M$. Then:
\begin{equation} \label{eq:CGammaBvsNoise}
	R_B(\gamma_{1})= C \left(\frac{\gamma_B}{1+\gamma_{1}} \right) \stackrel{\mathrm{(a)}}{>}  C \left(\frac{\gamma_B}{1+\Gamma_{M}} \right)=C(\Gamma^{JD}_B)
\end{equation}
where (a) follows from $\gamma_{1} < \Gamma_M$.

Now let $\gamma_{1} \geq \Gamma_M$. The  $\gamma_{1}>0$ implies:
\begin{equation}\label{eq:CGammaBvsNoise1}
	C \left(\frac{\gamma_B}{1+\gamma_{1}} \right) < C(\gamma_B)
\end{equation}
Furthermore, 
\begin{align}\label{eq:CGammaB-GammaMvsNoise}
 \C{\gamma_{1} + \gamma_B} -C(\Gamma_M)  &= \log_2 \left( \frac{1+\gamma_{1} + \gamma_B}{1+\Gamma_M} \right) \\ \nonumber
																				 &\stackrel{\mathrm{(a)}}{\geq} \log_2 \left(1+ \frac{\gamma_B}{1+\Gamma_M} \right)\\ \nonumber
																				 &\stackrel{\mathrm{(b)}}{\geq} \log_2 \left(1+ \frac{\gamma_B}{1+\gamma_{1}} \right) 
\end{align}
where both (a) and (b) follow from $\gamma_{1} \geq \Gamma_M$. From (\ref{eq:CGammaBvsNoise1}) and (b) in (\ref{eq:CGammaB-GammaMvsNoise}) it follows that:
\begin{align}
	&\max\left\{C \left(\frac{\gamma_B}{1+\gamma_{1}} \right), \min \left\{ \C{\gamma_B},  \C{\gamma_{1} + \gamma_B} -C(\Gamma_M)\right\}  \right\}\\ \nonumber
	&=\min \left\{ \C{\gamma_B},  \C{\gamma_{1} + \gamma_B} -C(\Gamma_M)\right\} 
\end{align}
while from $C(\gamma_B) > C \left(\frac{\gamma_B}{1+\Gamma_M} \right)$ and (a) in (\ref{eq:CGammaB-GammaMvsNoise}) it follows that:
\begin{align}
	C(\Gamma^{JD}_B)=C \left(\frac{\gamma_B}{1+\Gamma_M} \right) \leq \min \{ \C{\gamma_B},  \C{\gamma_{1} + \gamma_B} -C(\Gamma_M)\} 
\end{align}
We have thus proved that $C(\Gamma^{JD}_B)$ is always decodable, regardless of the value of $\gamma_{1}$. On the other hand, from (a) in (\ref{eq:CGammaB-GammaMvsNoise}) it follows that 
\begin{equation}
	R_B(\Gamma_M)=C(\Gamma^{JD}_B)
\end{equation}
such that $C(\Gamma^{JD}_B)$ is the minimal value of $R_B(\Gamma_M)$ in the interval $\gamma_{1}>0$. It is easily checked that $R_B(\Gamma_M)$ is continuous, since if $\gamma_{1}$ is treated as noise, the achievable rate for $\gamma_{1}=\Gamma_M$ is $C(\Gamma^{JD}_B)$.
\end{proof}

A final remark is in order. The previous lemma shows that it is possible for $U$ to decode the signal from $B$, but it only indirectly focuses on \emph{how} $U$ does it. We refer again to the three regions on Fig.~\ref{fig:RateRegionJD_PoB}. When $\gamma_{1}< \Gamma_M$, $U$ treats the signal from $M_1$ as noise. When $\Gamma_M \leq \gamma_{1}< \Gamma_M(1+\gamma_B)$, $U$ jointly decodes the signals from $B$ and $M_1$ and this is arguably the most complex regime of operation; we will therefore consider a simplified receiver in Section~\ref{sec:SuccessiveUserDecoding}. Finally, when $\gamma_{1} \geq \Gamma_M(1+\gamma_B)$, $U$ decodes first the signal from $M_1$, subtracts it, and proceeds to decode the signal of $B$ using single-user decoding. Note that it is not necessary for $B$ to be aware in which decoding regime does $U$ operate, it only matters that $U$ can decode $R_B$.

\subsection{Multiple MTDs: $N_M>1$}
\label{sec:MultipleMTDsNM1}

When multiple MTDs are associated with a single device $U$, then $U$ observes a multiple access channel of $N_M+1>2$ users. Nevertheless, we can use the insights of the proof for the case $N_M=1$ in order to state and prove the maximal downlink rate $R_B$ that can always be decoded by $U$, irrespective of the instantaneous channel gains $\gamma_{1}, \gamma_{2}, \cdots \gamma_{N_M}$. 

In order to provide an intuition, let us assume that there are $N_M=2$ MTDs with $\gamma_{1}$ and $\gamma_{2}$, respectively. Both MTDs transmit by using the same rate $R_M=C(\Gamma_M)$. If both SNRs are very low, such that $\gamma_{1}<\Gamma_M$ and $\gamma_{2}<\Gamma_M$, then both signals should be treated as noise, such that the maximal achievable rate for $R_B$ would be
\begin{equation}
C \left(\frac{\gamma_B}{1+\gamma_{1}+\gamma_{2}} \right) \label{eq:GammaM1M2treatedasnoise}
\end{equation}
Let both SNRs increase and reach $\gamma_{1}=\gamma_{2}=\Gamma_M$. Then one might think, following the argument in the proof of Lemma~\ref{prop:GammaB_JUD}, that this is the point where both MTD signals become decodable. This would lead to a rate $R_B=C\left(\gamma_B/(1+2\Gamma_M) \right)$ which should be the desired rate with zero outage. However, this argument is flawed. The reason is that in the multiple access channel consisting of $M_1$ and $M_2$ there are three inequalities that need to be satisfied in order to have both signals decoded. The third inequality is:
\begin{equation} \label{eq:ThirdInequalityM1M2}
	R_M+R_M \leq \log_2(1+\gamma_{1}+\gamma_{2})
\end{equation}
The reader can easily check that this is violated when $\gamma_{1}=\gamma_{2}=\Gamma_M$. Hence, we need to look for the minimal value of $R_B$ by considering the joint decodability of both signals. In order to satisfy (\ref{eq:ThirdInequalityM1M2}) with equality, the following needs to hold:
\begin{equation}
	1+\gamma_{1}+\gamma_{2}=(1+\Gamma_M)^2
\end{equation}
Thus, if we set $\gamma_{1}=\gamma_{2}=\frac{(1+\Gamma_M)^2-1}{2}$, then all three inequalities of the multiple access channel (with $M_1$ and $M_2$ as transmitters) are satisfied. This leads us to conjecture that the maximal downlink rate that has zero probability of outage is specified by
\begin{equation}
	\Gamma_{B,2}= \frac{\gamma_B}{(1+\Gamma_M)^2}
\end{equation}
This conjecture is proved in its general form through the following theorem:
\begin{thm}\label{thm:GammaB_NM>1}
Let there be $N_M$ MTDs, each sending at rate $R_M$ and $\gamma_{i}$ is the SNR of the link $M_i-U$. Let $B$ know $R_M$ and the SNR $\gamma_B$, but none of the SNRs $\gamma_{1}, \gamma_{2}, \ldots \gamma_{N_M}$. Then the maximal downlink transmission rate that is always decodable by $U$ is $R_{B,N_M} = C(\Gamma^{JD}_{B,N_M})$ where
\begin{equation}\label{eq:GammaBNM>1Users}
	\Gamma^{JD}_{B,N_M}= \frac{\gamma_B}{(1 + \Gamma_M)^{N_M}}
\end{equation}
and $\Gamma_M=C^{-1}(R_M)$.  
\end{thm} 
\begin{proof}
The proof is given in Appendix~\ref{sec:ProofOfGammaB_NM>1}.
\end{proof}

An interesting side effect of the derived zero outage downlink rate bounds is how they affect the outage of the MTDs links towards $U$. In the case of joint decoding (JD), all inequalities that involve $R_B$ must be satisfied, due to the assumed decodability of $R_B$, such that it follows that the decodability of each individual $M_i$ does not depend on the transmission from $B$. In other words, we can calculate the probability of outage for the MTDs by considering only the multiple access channel of $N_M$ users, which contains only the MTDs as transmitters, but not $B$.

\subsection{Single-User Decoding}
\label{sec:SuccessiveUserDecoding}

We now consider the case in which $U$ applies Single-User Decoding (SD) in each step, treating the yet-to-be-decoded user as a noise. 
\begin{figure}[tb]
	\begin{center}
		\includegraphics[width=\linewidth]{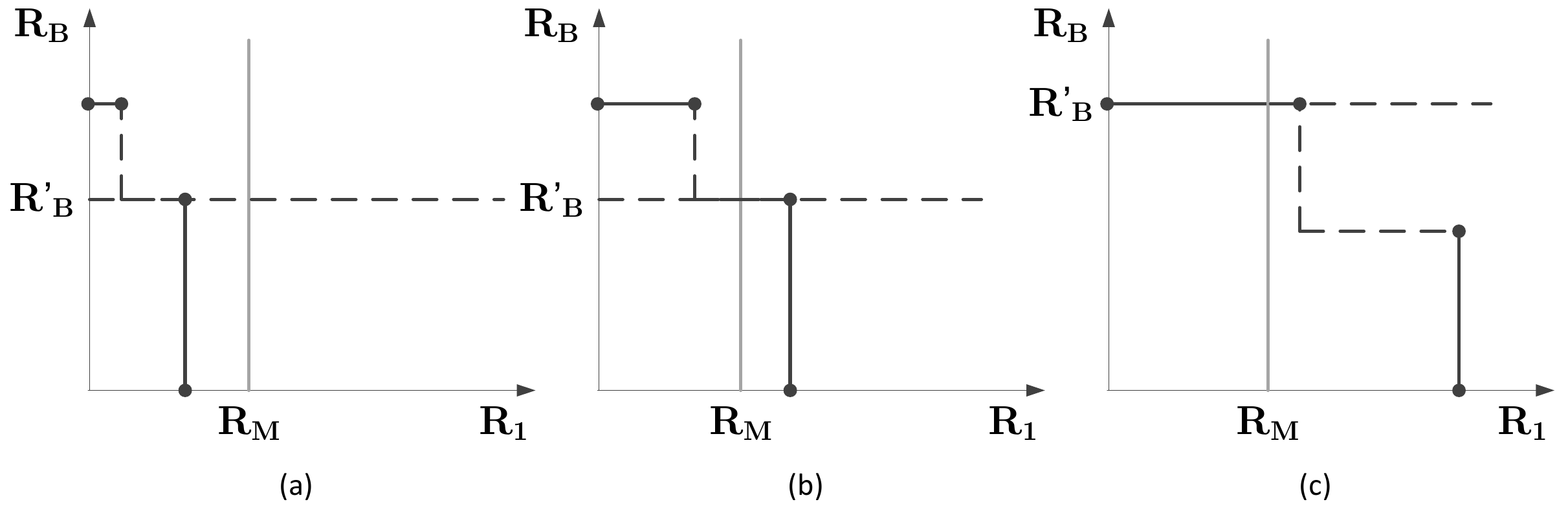}
	\end{center}
	\caption{Illustration of the working regimes of a Single-User decoder with successive interference cancellation, where one of the links has fixed rate. The transmission from $B$: (a) is decoded while treating the MTD transmission as noise; (b) is decoded while treating the MTD transmission as noise, due to the Single-User decoding; (c) is decoded in the presence of noise, after the transmission from the MTD has been decoded and subtracted from the the composite signal.}
	\label{fig:SDRegime}
\end{figure}

We treat first the case $N_M=1$, where we show how to select the maximal downlink rate for given $\gamma_B$ that guarantees that there is no outage. Recall that there are three operating regions of the JD receiver for $N_M = 1$. When SD is applied, there are only two decoding options, depicted on Fig.~\ref{fig:SDRegime}:
\begin{itemize}
\item If the link $M_1-U$ is weak, then $\gamma_1$ is low and the signal from $B$ should be decoded by treating the signal from $M_1$ as noise. This regime is depicted in Fig.~\ref{fig:SDRegime}(a) and Fig.~\ref{fig:SDRegime}(b). 
\item If the link $M_1-U$ is very strong, given that the MTD rate is fixed to $R_M$, then the signal from $M_1$ is decoded, subtracted and a ``clean'' signal of $B$ is decoded at the maximal possible rate $C(\gamma_B)$. It can be shown that the signal from $M_1$ is decodable and thus the link $M_1-U$ can be treated as strong when $\gamma_1 = \phi_1 = \Gamma_M(1+\gamma_B)$. This regime is depicted in Fig.~\ref{fig:SDRegime}(c).
\end{itemize}

We can then define $R_B (\gamma_{1})$ as the maximal achievable $R_B$ for given $\gamma_{1}$, which can be written as:
\begin{equation}\label{eq:RBcompactSD}
	R_B(\gamma_{1}) = \left\{
  \begin{array}{lr}
    \C{\frac{\gamma_B}{1+\gamma_{1}}} &  \gamma_{1} <  \Gamma_M(1+\gamma_B)\\
    \C{\gamma_B}                                    &  \gamma_{1} \geq \Gamma_M(1+\gamma_B) \\
  \end{array}
\right.
\end{equation}

\begin{prop}\label{lem:GammaB_SUDNe1}
Let there be a single MTD sending at rate $R_M$ and let $U$ apply successive single user decoding. If $B$ knows the instantaneous SNR $\gamma_B$, but not $\gamma_{1}$, then the maximal downlink transmission rate that is always decodable by $U$ is $R_B = C(\Gamma^{SD}_B)$ where
	\begin{equation}\label{eq:GammaBSingleUserSD}
		\Gamma^{SD}_B = \frac{\gamma_B}{1+\Gamma_M(1 + \gamma_B)}
	\end{equation}
	and $\Gamma_M=C^{-1}(R_M)$.  
\end{prop}
\begin{proof}
\begin{figure}
	\centering
		\includegraphics[width=0.8\linewidth]{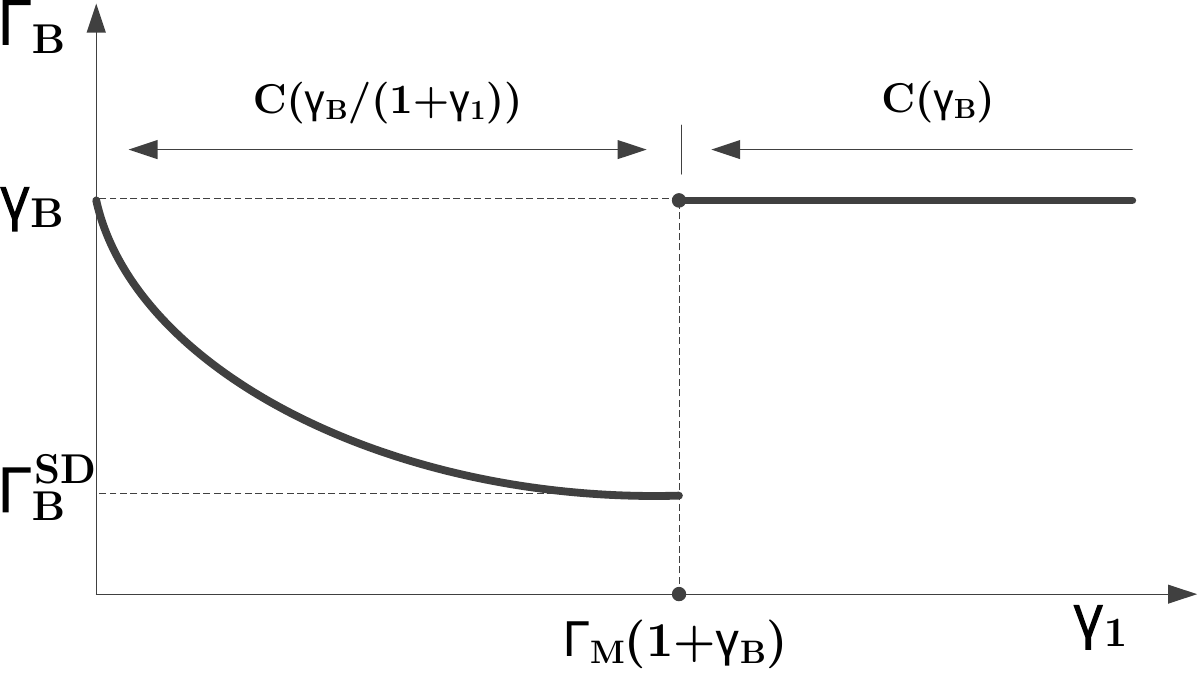}
	\caption{Achievable $\Gamma_B$ when using Single-User Decoding (SD) given known $\gamma_B$ that does not put the $B-U$ link in outage.}
	\label{fig:RateRegionSD_PoB}
\end{figure}

The proof is straightforward. If $\gamma_1<\Gamma_M(1+\gamma_B)$ then:
\begin{equation}
\Gamma^{SD}_B= \frac{\gamma_B}{1+\Gamma_M(1 + \gamma_B)}< \frac{\gamma_B}{1+\gamma_1}
\end{equation}
If  $\gamma_{1} \geq \Gamma_M(1+\gamma_B)$, then 
\begin{equation}
\Gamma^{SD}_B< \gamma_B
\end{equation}
since $\Gamma_M(1+\gamma_B)>0$. The function $R_B(\gamma_{1})$ plotted on Fig.~\ref{fig:RateRegionSD_PoB} approaches value $C(\Gamma^{SD}_B)$ as $\gamma_1=\Gamma_M(1 + \gamma_B)-\epsilon$, where $\epsilon>0$ and $\epsilon \rightarrow 0$. Stated precisely, $C(\Gamma^{SD}_B)$ is the supremum of the zero-outage $R_B$.
\end{proof}
Clearly, the zero outage rate in the SD setting is much lower than the one in the JD setting.
Analogous to the case of JD, it is natural to ask what happens when $N_M>1$ and still assuming that $B$ is ignorant about $\gamma_j$, where $j=1, 2, \ldots N_M$. The following proposition shows that there is no positive rate $R_B$ that can lead to zero outage when SD is applied. 

\begin{prop}\label{lem:GammaB_SUDNg1}
Let there be a $N_M>1$ MTDs, each sending at rate $R_M$ and let $U$ apply successive single user decoding. If $B$ knows the instantaneous SNR $\gamma_B$, but not the SNRs of the links MTD-$U$, then there is no positive downlink rate that can guarantee that $U$ can decode the signal of $B$, i.e.
\begin{equation}
\Gamma^{SD}_{B,N_M}=0 \qquad \mathrm{ for } \qquad N_M>1
\end{equation}
\end{prop}
\begin{proof}
It is sufficient to show this for the case $N_M=2$ and consider $\gamma_1$ and $\gamma_2$, assuming without loss of generality that $\gamma_1 \leq \gamma_2$. Since single-user decoding is used, the signal of $M_2$ should be decoded first, treating the signals of $M_1$ and $B$ as a noise. Let us consider only the case 
$\gamma_2< \frac{\Gamma_M}{1+\gamma_1+\gamma_B}$. Here the signal from $M_2$ cannot be decoded and should be treated as noise. But, since 
\begin{equation}
\gamma_1 \leq \gamma_2< \frac{\Gamma_M}{1+\gamma_1+\gamma_B}
\end{equation}
it follows that both the signals from $M_1$ and $M_2$ should be treated as noise. However, in that case the rate that guarantees no outage of the downlink signal should be chosen:
\begin{equation}
\Gamma^{SD}_{B,2} \leq  \frac{\gamma_B}{1+\gamma_1+\gamma_2}
\end{equation}
which goes arbitrarily close to zero as $\gamma_1, \gamma_2$ increase. Thus the maximal zero-outage rate is zero. Since $\Gamma_{B,N_M}^{JD}$ does not increase with $N_M$, it follows that the claim of the proposition is true for any $N_M>1$. 
\end{proof}
The last proposition reveals that joint decoding is instrumental in preserving a nonzero rate when the interfering channels $\gamma_j$ are unknown.

\subsection{Practical Considerations} 
\label{sub:practical_considerations}

The derived zero-outage rates for both JD and SD, show that $B$ is not required to have CSIT from the MTD-U links. Therefore, this information does not need to be exchanged prior to connection establishment, neither during the time while the communications are active. Nevertheless, there are other signaling exchanges that need to be accounted for, as we discuss in the following.

First, it is the time synchronization, which can be accomplished by each device listening to the downlink control signaling and possibly some other timing adjustments. Depending on the synchronization schemes in place it might correspond to extra cooperation overhead. Second, the cellular receiver needs to be able to estimate the channels of the active transmitters over the MAC, i.e. the channels $M_q-U$ of the transmitting MTDs and the channel $B-U$. This can be achieved by having dedicated preambles for each MTD that is associated with a given cellular receiver $U$, as well as a dedicated preamble used by $B$ for downlink transmission. 

Since all MTDs use a single codebook that is known both at the cellular user $U$ as well as at the base station $B$, there is no need to operationally exchange information on the MTD codebooks. An open question remains how to design a practical joint decoding scheme where one or multiple (when $N_M>1$) links have a fixed rate. We reiterate the fact that the existence of single MTD codebook enables efficient design of the joint decoder. 

Finally, it should be noted that zero-outage downlink transmission is possible because the interference power that is subject to variation comes from the devices that send at fixed rates and whose signals can be decoded if their interference is sufficiently strong. However, if there is an external source of variable interference that is not decodable, regardless of how large the interference power is, then it is not possible to have zero-outage downlink transmissions. The reason is that, to achieve zero outage, the receiver needs to know the instantaneous MAC region while that is not the case if there is an external variable interference. In practice, the receiver can have an estimate upper bound on the power coming from external interference and operate with a reduced (pessimistic) version of the MAC region.


\section{Performance Characterization}
\label{sec:PerformanceAnalysis}

In this section we derive the expected value of the normalized capacity in the $B-U$ link, $E[R_B]$, resultant from the zero outage rates for both decoding settings at $U$.
These analytical results are compared numerically to stochastic simulations in Section~\ref{sec:Numerical Results}.

\subsection{Preliminaries}
\label{sec:Preliminaries}

The analytical results are obtained in a stochastic geometry setting.
Therefore, without loss of generality, we assume that all network nodes are deployed uniformly and isotropically in a disk with radius $R$, with $U$ positioned at its origin.
The considered networks nodes are: one Base Station $B$, $N_M$ MTDs associated with $U$ and $N_I$ MTDs not associated with $U$.

Further, for analytical tractability we consider the following bounded pathloss and fading model for the network nodes\footnote{We note that $|h|^2$ is bounded, since it avoids the singularity at $r=0$~\cite{5226961}, allowing the resulting SNR distributions to have finite moments~\cite{Haenggi2012}.},
\begin{equation}\label{eq:PathLossFading}
	|h_x|^2 = h K (1 + r_x)^{-\alpha}
\end{equation}
where $\alpha$ is the path-loss exponent, $K_x$ denotes the extra losses associated with node $x$, $r_x$ is the distance between the node $x$ and $U$ and $h$ is the instantaneous fading realization with distribution $f_{h}(x)$.

The expected value of the received power of the signal transmitted by the $x$ node, $E[\zeta_x]$, is computed as follows,
\begin{align}\label{eq:expectedReceivedPower}
	E[\zeta_x]  &= P_x E[|h_x|^2]\\ \nonumber
							&= \int_0^R \int_0^\infty K P_x  v (1 + u)^{-\alpha} f_{h}(v) f_{R}(u) \mathrm{d}v \mathrm{d}u\\ \nonumber
	            &\overset{(a)}{=} \int_0^R K P_x (1 + u)^{-\alpha} f_{R}(u) \mathrm{d}u\\ \nonumber
							&\overset{(b)}{=} -\frac{2 K P_x (R \alpha + R^2 \alpha - (R + 1)^\alpha - R^2 + 1)}{R^2 (R + 1)^\alpha (\alpha^2 - 3\alpha + 2)}
\end{align}
where (a) follows from $f_{h}(x)$ being the exponential fading distribution with unitary mean, i.e. $E[|h|^2]=1$, as stated in Section~\ref{sec:SystemModel}.
(b) follows from assuming that $\alpha>2$ and that $f_{R}(x)$ is the distribution that models the distance between the $U$ and $x$, which is derived by noticing that the uniform and isotropic distribution of nodes in a disk is proportional to the arc of the disk edge,
\begin{equation*}
		f_R(r) = \rho 2 \pi r \Longleftrightarrow \int_0^R f_R(u) \mathrm{d}u = 1 \Longleftrightarrow \rho = \frac{1}{\pi R^2}.
\end{equation*}
The resulting pdf is $f_{R}(r) = 2r / R^2$ and the corresponding cdf $F_{R}(r) = r^2 / R^2$.

To model the interference contribution of the $N_I$ MTDs not associated with $U$, we recall $\tilde{z}$ from~\eqref{y_u}, which we define:
\begin{equation} \label{z_noise_and_interference}
	\tilde{z} = \sum_{k=1}^{N_I} h_{M_k} x_{M_k} + z
\end{equation}
where $z$ denotes the complex Gaussian noise with variance $E[|z|^2] = \sigma^2$, $h_{M_k}$ denote the respective channel gains from the $M_k-U$ link and $x_{M_k}$ denotes the zero mean Gaussian complex signal with $E[|x_{M_k}|^2] = P_M$.
The variance $E[|\tilde{z}|^2]$ is computed as follows,
\begin{align}
	E[|\tilde{z}|^2] &= \tilde{\sigma}^2 = P_M \sum_k^{N_I} |h_{M_k}|^2 + E[|z|^2] = \sum_{k=1}^{N_I} E[\zeta_k] + E[|z|^2] \\ \nonumber
									 &= -N_I\frac{2 K P_M (R \alpha + R^2 \alpha - (R + 1)^\alpha - R^2 + 1)}{R^2 (R + 1)^\alpha (\alpha^2 - 3\alpha + 2)} + \sigma^2\\ \nonumber
									 &\overset{(a)}{=} -N_I\frac{2 K P_M (\sqrt{\frac{N_I}{\pi \lambda_I}} \alpha + \sqrt{\frac{N_I}{\pi \lambda_I}}^2 \alpha - (\sqrt{\frac{N_I}{\pi \lambda_I}} + 1)^\alpha - \sqrt{\frac{N_I}{\pi \lambda_I}}^2 + 1)}{\sqrt{\frac{N_I}{\pi \lambda_I}}^2 (\sqrt{\frac{N_I}{\pi \lambda_I}} + 1)^\alpha (\alpha^2 - 3\alpha + 2)} + \sigma^2.
\end{align}
where (a) comes from the $N_I$ nodes being deployed uniformly and isotropically in the disk with radius $R$ and area $\pi R^2$.
Then, the interfering MTDs density is given by $\lambda_I = \frac{N_I}{\pi R^2}$.
The asymptotic case is computed by applying the substitution $R = \sqrt{\frac{N_I}{\pi \lambda_I}}$ and taking the limits when $N_I$ tends to $\infty$, with constant $\lambda_I$. Then we obtain the following asymptotic approximation.
\begin{equation}
	\lim_{N_I \rightarrow \infty} E[|\tilde{z}|^2] = \frac{2 \pi K P_M \lambda_I}{\left(\alpha^2 - 3\alpha + 2\right)} + \sigma^2
\end{equation}
We note that this approximation is very close to $E[|\tilde{z}|^2]$ with finite $R$.

Finally, we lower-bound  $E[\gamma_x]$ using the Jensen's inequality:
\begin{align}\label{eq:ExpectedSNR}
	E[\gamma_x] & = E[ \frac{\zeta_x} {\sum_{k=1}^{N_I} E[\zeta_k] + \sigma^2 } ] = E[ \zeta_x] E[ \frac{1} {\sum_{k=1}^{N_I} E[\zeta_k] + \sigma^2 } ] \\ \nonumber
				&\geq \frac{E[\zeta_x]}{E[\sum_{k=1}^{N_I} E[\zeta_k] + \sigma^2 ]} 
				= \frac {E[\zeta_x]} {\tilde{\sigma}^2}
\end{align}
Where the lower bound is due to $\frac{1}{x+a}$ being convex in the domain of $x$, when $a>0$. If there is no interference, then there is equality in~(\ref{eq:ExpectedSNR}). 

\subsection{Mean Downlink Rate}
\label{sec:DownlinkRate}

In the JD setting, the mean downlink rate for $N_M$ connected MTDs, $E[R_B^{JD}]$, is computed as follows:
\begin{align}\label{eq:GammaJDB}
	E[R_B^{JD}] &\overset{(a)}{\leq} \C{ E[\Gamma_B^{JD}]}\\ \nonumber
							&= \C{\int_0^{\infty} \frac{u}{\left(1 + \Gamma_M\right)^{N_M}} f_{\gamma_B}(u)  \mathrm{d}u}\\ \nonumber
							&= \C{\frac{E[\gamma_B]}{\left(1 + \Gamma_M\right)^{N_M}}}. 
\end{align}
where $f_{\gamma_B}(x)$ denotes the $\gamma_B$ distribution.
The (a) upper bound results from the Jensen's inequality since $\C{x}$ is concave in the domain of $x$. If there is no interference, then substituting~(\ref{eq:ExpectedSNR}) for $E[\gamma_B]$ preserves the upper bound. However, if there is interference, then~(\ref{eq:ExpectedSNR}) provides a lower bound, such that one cannot claim the upper bound in~(\ref{eq:GammaJDB}). Therefore we use the approximation:
\begin{align}\label{eq:GammaJDBapprox}
	E[R_B^{JD}] \approx \C{\frac{E[\zeta_B]}{\tilde{\sigma}^2 \left(1 + \Gamma_M\right)^{N_M}}}. 
\end{align}
In the SD setting, the downlink zero outage upper bound is only provided for $N_M = 1$.
The mean downlink rate, $E[R_B^{SD}]$, is computed as follows,
\begin{align}\label{eq:GammaSD_B}
	E[R_B^{SD}] &\leq \C{ E[\Gamma_B^{SD}]}\\ \nonumber
							&= \C{\int_0^{\infty} \frac{u}{1 + \Gamma_M\left(1 + u\right)} f_{\gamma_B}(u)  \mathrm{d}u}\\ \nonumber
							&\overset{(a)}{\leq} \C{\frac{E[\gamma_B]}{1 + \Gamma_M\left(1 + E[\gamma_B]\right)}} \\ \nonumber
							&\overset{(b)}{\approx} \C{\frac{E[\zeta_B]}{\tilde{\sigma}^2 (1 + \Gamma_M\left(1 + E[\zeta_B]/\tilde{\sigma}^2\right))}}.
\end{align}
where (a) the upper bound is given by the Jensen's inequality since the function $g(u) = \frac{u}{a + b u}$ is strictly concave in the domain of u, for $a, b \in \Re^+$. We again use the approximation (b), similar to the discussion about~(\ref{eq:GammaJDBapprox}).

\section{Numerical Results}
\label{sec:Numerical Results}

In the following we show the numerical results obtained by evaluation of the analytical expressions and through stochastic simulations with enough repetitions to ensure numerical stability.
Table~\ref{tab:SimulationScenarioSettings} lists the more relevant considered system parameters.
\begin{table}[t]
	\centering
		\begin{tabular}{ l c }
		\hline
		\textbf{Parameter} & \textbf{Value} \\ \hline
		$\sigma^2$ & $ -97.5 [dBm]$~\cite{HarriHolma2011} \\
		$\alpha$ & 4 \\
		$P_M$ & $-10 [dBm]$ \\
		$P_B$ & $30 [dBm]$ \\
		$K$ & $-30 [dB]$ \\
		$R$ & $200 [m]$ \\
		\hline
		\end{tabular}
	\caption{Simulation scenario settings.}
	\label{tab:SimulationScenarioSettings}
\end{table}
%


We consider the case where there is a single MTD associated with $U$ and no aggregate interference, i.e. $N_M = 1$ and $N_I = 0$. In absence of external interference, the analytical expressions $E[R^{JD}_B]$ and $E[R^{SD}_B]$ serve as upper bounds, as depicted on Fig.~\ref{fig:NumResNMe1NIe0}. It can be seen that the upper bound for JD is tight, while it is quite loose for SD. We observe that the cost of using the lower complexity receiver architecture SD is a lower zero outage downlink rate $E[R_B]$, than the one achievable with JD. Nevertheless, as $R_M$ tends to very low values, the average rates $E[R^{JD}_B]$ and $E[R^{SD}_B]$ converge to the value obtained as if the machine-type transmission is absent. 

\begin{figure}
	\centering
		\includegraphics[width=\linewidth]{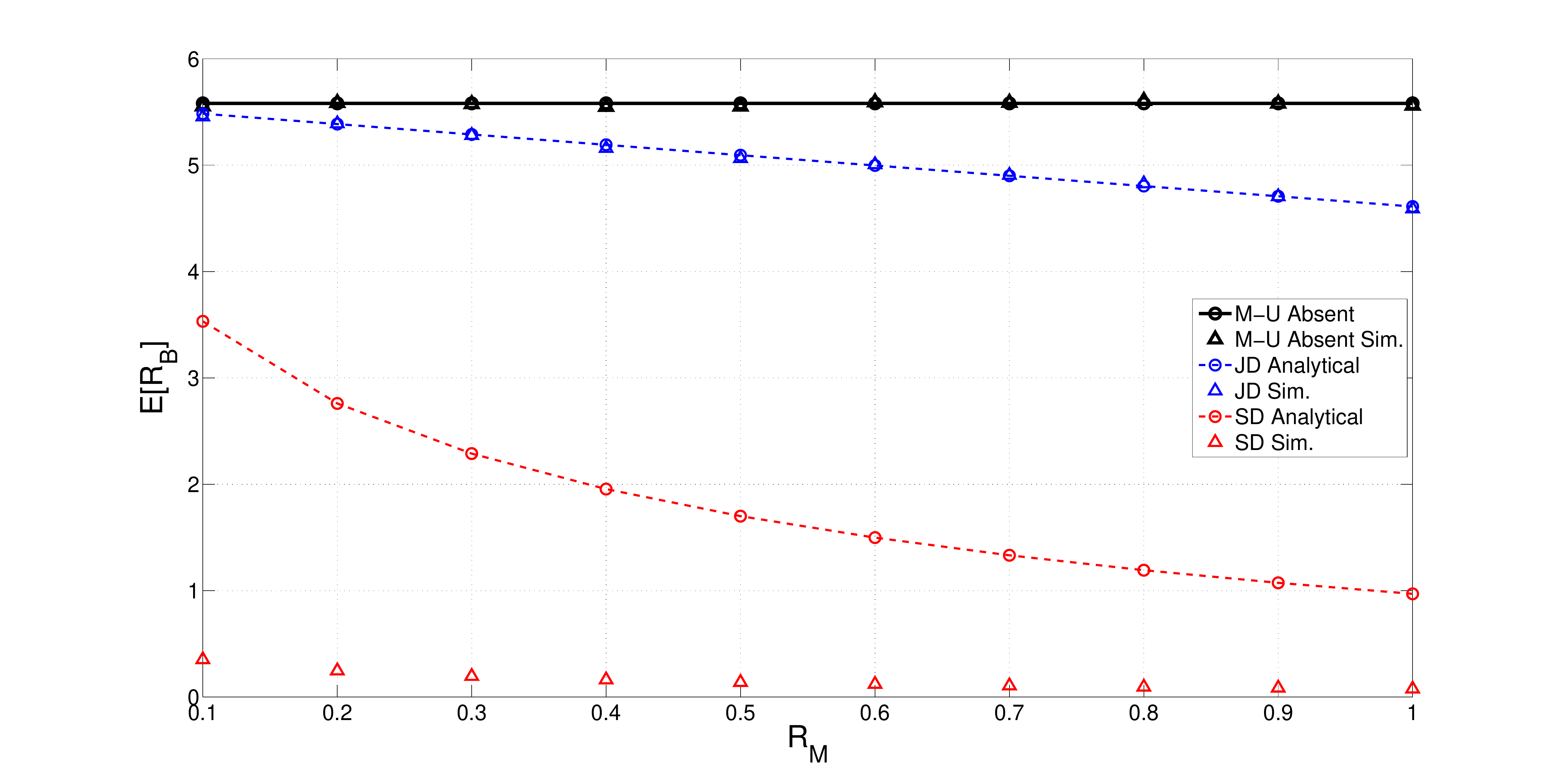}
	\caption{Comparison of the $E[R_B]$ derived analytical upper bounds with the simulations for the JD and SD settings, when $N_M = 1$ and $N_I = 0$.}
	\label{fig:NumResNMe1NIe0}
\end{figure}
%


We now consider the case where there are multiple MTDs associated with $U$, i.e. $N_M > 1$, while still ignoring the external interference. The selection of a zero outage downlink rate is considered only in the JD setting, since, as elaborated previously, it is impossible to have it when SD is used. In Fig.~\ref{fig:NumericalResultsNMg1NMI0} is depicted the behavior of zero outage downlink rate with the number of users. As expected, the higher is the number of users and $R_M$ then the lower is the zero-outage rate $R_B$.
Of special interest, is the case where $R_M$ is very low, since from the results it can be seen that a large number of MTDs can be supported with minimal impact on the allowed $E[R_B]$.
\begin{figure}
	\centering
		\includegraphics[width=\linewidth]{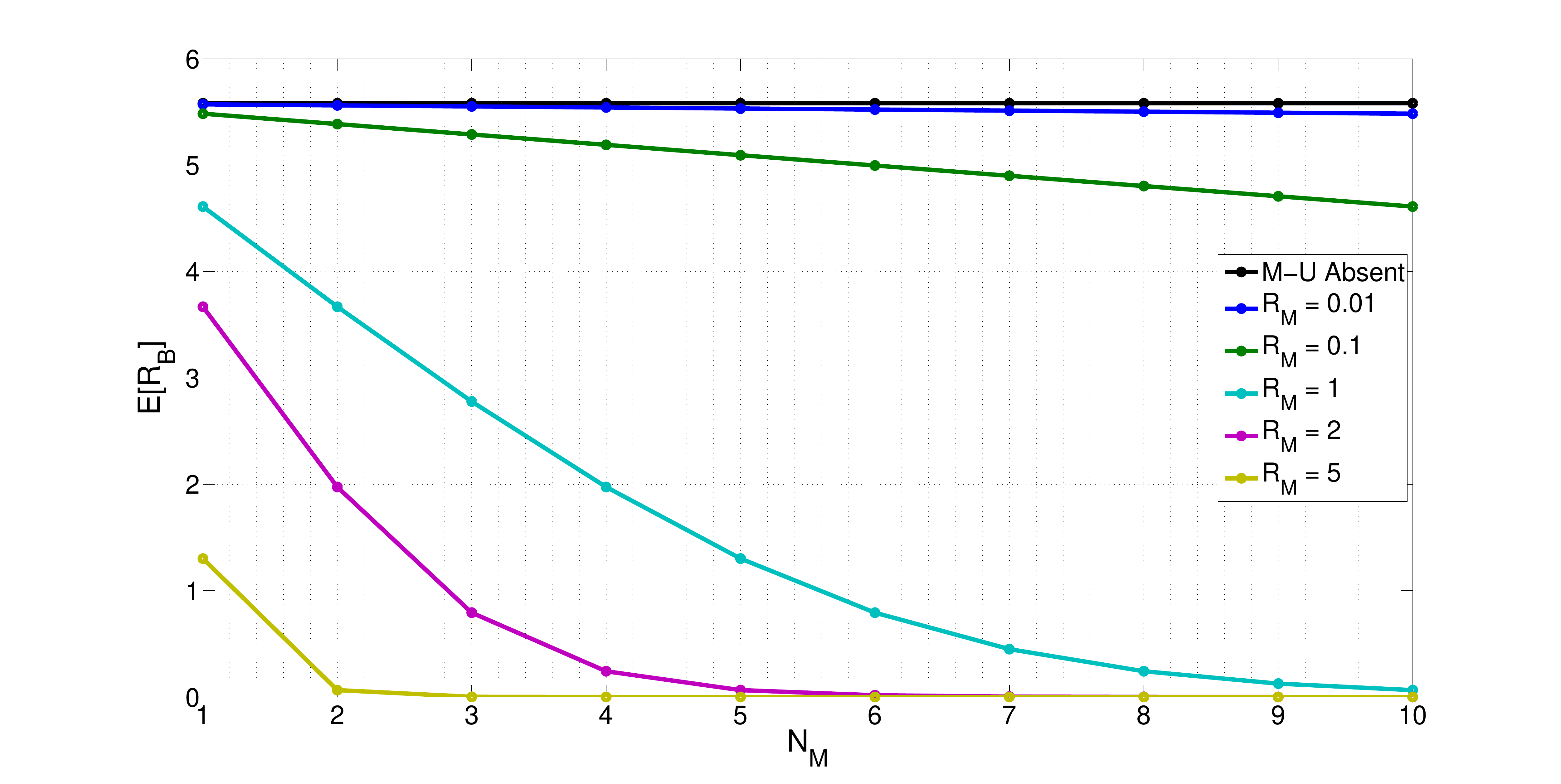}
	\caption{$E[R_B]$ for multiple MTDs in the JD decoding setting, using the analytical upper bound defined in~\eqref{eq:GammaJDB}.}
	\label{fig:NumericalResultsNMg1NMI0}
\end{figure}
%


Finally, we consider the case $N_M = 1$ and aggregate interference due to other MTDs not associated with $U$, i.e. $N_I > 1$.
Fig.~\ref{fig:NumericalResultsNMg1Ratio} depicts the behavior of analytical approximation $E[R_B]$ (which is now an approximation and not an upper bound) with increasing $\lambda_I$, i.e. increasing the aggregated interference. We note that the analytical approximations while loose, still capture the two main trends observed in the simulation results. 
\begin{figure}
	\centering
		\includegraphics[width=\linewidth]{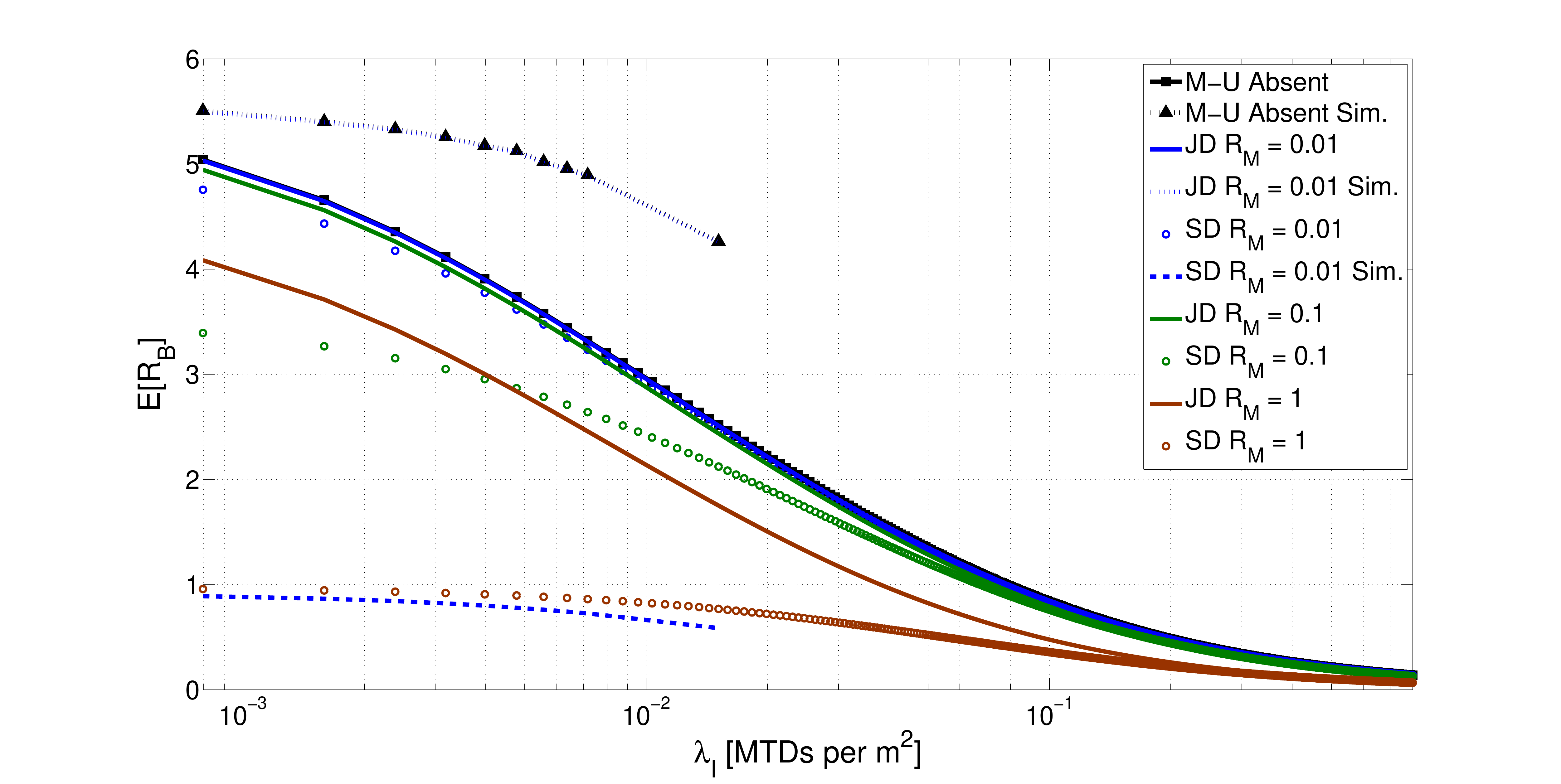}
	\caption{$E[R_B]$ of JD and SD versus $\lambda_I$, where simulation results are provided for comparison only for lower $\lambda_I$ and $R_M = 0.01$, due to the high computational complexity associated with higher densities and, in the case of $R_M = 0.01$, to not overcrowd the plot.}
	\label{fig:NumericalResultsNMg1Ratio}
\end{figure}
First, as the aggregate interference $\lambda_I$ increases, the downlink rate $E[R_B]$ decreases. Furthermore, the rates $E[R^{JD}_B]$ and $E[R^{SD}_B]$ converge, since in the regime of high interference, the MAC region at the cellular receiver shrinks and JD no longer provides benefits over SD. Second, for very small values of $R_M$ the average downlink rate $E[R_B]$ becomes very close to the one achievable when the transmission from $M_1$ is absent, both for JD and SD. 

\section{Conclusion}
\label{sec:Conclusion}

We have considered a network-assisted Device-to-Device (D2D) scenario that enables the underlay of fixed-rate links in the cellular downlink. The motivation for considering fixed-rate D2D links is rooted in Machine-Type Communications (MTC), also known as M2M communications, which feature fixed, low rates. We consider a scenario in which a user $U$ receives a downlink transmission from a Base Station $B$, while simultaneously receiving signals from $N_M>1$ Machine-Type Devices (MTDs) attached to $U$ and each using a fixed rate $R_M$. Thus, contrary to the mainstream, the underlay operation of the D2D links in our scenario is supported during the downlink transmission. Assuming that $B$ knows only the channel $B-U$, but not the interfering channels from the MTDs to $U$, we have proven that there is a positive downlink rate that can always be decoded by $U$, leading to zero-outage of the downlink signal. Such a positive rate exists only when $U$ uses the full decoding region of the multiple access channel defined by $N_M+1$ transmitters: $B$ and $N_M$ MTDs. We have provided a close form of the maximal zero-outage rate. We have also considered a simpler operation at $U$, where $U$ uses a Single-User Decoding (SD) and successively removes the decoded users. We have shown that, using SD, it is possible to have a positive downlink rate with zero outage only when $N_M=1$, but not when $N_M>1$. Overall, our approach shows that underlaying during the downlink transmission is a viable alternative when the multiple access channels created in such a setting are fully utilized at the downlink cellular receiver. One of the most interesting issues for future research is how to design a practical joint decoder at the cellular receiver that takes advantage of the fact that the codebooks of the MTDs are fixed, while the number of active MTDs can vary. 

\appendices

\section{Proof of Theorem~\ref{thm:GammaB_NM>1}}
\label{sec:ProofOfGammaB_NM>1}

In order to prove the theorem, we introduce the following notation for a multiple access channel (MAC) of $k$ transmitters  $\{M_1, M_2, \ldots M_{k}\}$, each of them sending at a fixed rate $R_M=C(\Gamma_M)$.
For an easier notation, let $\gamma_i$ be the SNR of the device $M_i$. In its original formulation, the MAC problem requires all $k$ signals to be decoded successfully and is described by a set of $2^{k}-1$ inequalities. Let ${\cal S}$ denote any nonempty subset from the set of MTDs $\{M_1, M_2, \ldots M_{k}\}$ and $|{\cal S}|$ denote its cardinality. Then the MAC inequalities can compactly be written as:
\begin{align}
	|{\cal S}|R_M  &\leq C \left( \sum_{M_i \in {\cal S}} \gamma_i \right) \label{eq:MACinequalityforS}
\end{align}
or in alternative form as:
\begin{align}
	1+\sum_{M_i \in {\cal S}} \gamma_i & \geq (1+\Gamma_M)^{|{\cal S}|} \label{eq:MACinequalityforSalt}
\end{align}
From the perspective of the selection of the maximal downlink rate $R_B$, it is not important that all $k$ signals from the MTD are decoded. We therefore reformulate the MAC problem and, for given $R_M$ and given set of SNRs $\{\gamma_i \}$, determine what is the maximal number of rate $R_M$ signals that can be decoded. 
We state and prove the following lemma:
\begin{lem} \label{lem:MAC_MTDsonly}
Let ${\cal S}_m$ be the subset of $\{M_1, M_2, \ldots M_{k}\}$ that has the lowest cardinality $|{\cal S}_m|=S$ among the subsets ${\cal S}$ for which the inequality (\ref{eq:MACinequalityforSalt}) is not satisfied. Then none of the signals from the devices in ${\cal S}_m$ can be decoded. 
\end{lem}
\begin{proof}
Without losing generality we can assume that ${\cal S}_m= \{M_1, M_2, \ldots M_S \}$. From the properties of the MAC channel consisting of ${\cal S}_m= \{M_1, M_2, \ldots M_S \}$ it follows that not all the signals from the devices in ${\cal S}_m$ can be decoded. Without losing generality, let us assume that the signals of the devices from the set ${\cal S}_u=\{M_1, M_2, \ldots M_T\}$, with $T<S$ are not decoded and should be treated as noise. Then we observe a new MAC with $S-T$ transmitters ${\cal S}_d=\{M_{T+1}, M_{T+2}, \ldots M_S\}$, but with SNRs that are scaled in order to account for the fact that the signals from the devices in ${\cal S}_u$ are treated as noise:
\begin{equation}\label{eq:ScaledSNRSSu}
 \gamma_{i}^{\bar{\cal S}_u}=\frac{\gamma_i}{1+\sum_{j=1}^T \gamma_j} \qquad \textrm{ for } T<i\leq S
\end{equation}
Since all signals from ${\cal S}_d$ are decodable, it follows that all the $2^{S-T}-1$ inequalities for the MAC channel with SNRs determined according to (\ref{eq:ScaledSNRSSu}) should be satisfied. Specifically, the following inequality needs to be satisfied:
\begin{equation}
	1+\sum_{j=T+1}^S \gamma_{j}^{\bar{\cal S}_u} \geq (1+\Gamma_M)^{S-T}
\end{equation}
which can be rewritten as:
\begin{equation}
	1+\frac{\sum_{j=T+1}^S \gamma_{j}}{1+\sum_{i=1}^T \gamma_{i}}\geq (1+\Gamma_M)^{S-T}
\end{equation}
This can be transformed as follows:
\begin{equation}\label{eq:Contradiction}
	1+\sum_{j=1}^S \gamma_{j} \geq (1+\Gamma_M)^{S-T} (1+\sum_{i=1}^T \gamma_{i}) \stackrel{\textrm{(a)}}{\geq} (1+\Gamma_M)^{S-T} (1+\Gamma_M)^{T} \geq (1+\Gamma_M)^{S} 
\end{equation}
where (a) follows from the assumption that the inequality  (\ref{eq:MACinequalityforSalt}) is satisfied for the set ${\cal S}_u$, since ${\cal S}_u \subset {\cal S}_m$. However, the inequality (\ref{eq:Contradiction}) contradicts the assumption that (\ref{eq:MACinequalityforSalt}) is violated for the set ${\cal S}_m$. Therefore, no signal from the devices in ${\cal S}_m$ can be decoded.
\end{proof}

We now proceed to the proof of the theorem.

\begin{proof}
\emph{(Theorem~\ref{thm:GammaB_NM>1})} We carry out the proof by using induction. The case $N_M=1$ is proved in Lemma~\ref{prop:GammaB_JUD}. We will then assume that the expression for $\Gamma_{B,N_M}$ is valid for $N_M=k$ and prove that it holds for $N_M=k+1$. The MAC consists of $k+2$ users, with $k+1$ MTDs and $B$. The $2^{k+2}-1$ inequalities that describe the MAC can be written as:
\begin{align}
	&1+\sum_{M_i \in {\cal S}} \gamma_i \geq (1+\Gamma_M)^{|{\cal S}|}  \label{eq:ThmMACinequalityforS} \\
	&R_B \leq C \left( \gamma_B+\sum_{M_i \in {\cal S}} \gamma_i \right)-|{\cal S}|R_M \label{eq:ThmMACinequalityforSB} \\
	&R_B \leq C(\gamma_B) \label{eq:ThmMACinequalityforB} 
\end{align}
where ${\cal S}$ is any nonempty subset of $\{M_1, M_2, \ldots M_{k+1} \}$. 

Let ${\cal S}_u$ be the set of minimal cardinality $|{\cal S}_u|=S$ for which the inequality (\ref{eq:ThmMACinequalityforS}) is not satisfied. We need to consider two cases $S>0$ and $S=0$. 

\subsubsection{$S>0$}  Without losing generality, assume that the set ${\cal S}_u=\{M_1, M_2, \ldots M_S \}$. According to Lemma~\ref{lem:MAC_MTDsonly}, none of the signals from the devices in ${\cal S}_u$ is decodable. Therefore, those signals should be treated as noise and we can observe a channel with $k+1-S$ MTDs and $B$, where the SNRs are scaled as:
\begin{align}\label{eq:ScaledSNRsBM}
	\gamma_{B}^{\bar{\cal S}_u} &=\frac{\gamma_B}{1+\sum_{j=1}^S \gamma_j} \nonumber \\
	\gamma_{i}^{\bar{\cal S}_u} &=\frac{\gamma_i}{1+\sum_{j=1}^S \gamma_j}
\end{align}
Using the inductive assumption, we know that the maximal zero-outage $R_B$ can be chosen to have the equivalent SNR of:
\begin{equation}
	\Gamma_{B,k+1-S}=\frac{\gamma_{B}^{\bar{\cal S}_u}}{(1+\Gamma_M)^{k+1-S}}=\frac{\gamma_B}{(1+\sum_{j=1}^S \gamma_j)(1+\Gamma_M)^{k+1-S}} \stackrel{\textrm{(a)}}{>} \frac{\gamma_B}{(1+\Gamma_M)^{k+1}}
\end{equation}
where (a) follows from the assumption that (\ref{eq:ThmMACinequalityforS}) is not satisfied for the set ${\cal S}_u$, translating into $1+\sum_{j=1}^S \gamma_j < (1+\Gamma_M)^{S}$.

\subsubsection{$S=0$} Here the signals from all MTDs are decodable, if we observe the MAC channel consisting only of MTDs and excluding $B$. Since it is obvious that $\gamma_B>\frac{\gamma_B}{(1+\Gamma_M)^{k+1}}$ for any $\Gamma_M>0$, we need to show that:
\begin{equation} \label{eq:NeedToshowS=0}
	C \left( \gamma_B+\sum_{M_i \in {\cal S}} \gamma_i \right)-|{\cal S}|R_M \geq  C \left ( \frac{\gamma_B}{(1+\Gamma_M)^{k+1}} \right ) 
\end{equation}
for all possible nonempty subsets ${\cal S}$. The equation (\ref{eq:NeedToshowS=0}) can be equivalently written as follows:
\begin{equation}
	\log_2 \left( \frac{1+\gamma_B+\sum_{M_i \in {\cal S}} \gamma_i }{(1+\Gamma_M)^S} \right) \geq \log_2 \left( 1+ \frac{\gamma_B}{(1+\Gamma_M)^{k+1}} \right)
\end{equation}
We write only the argument of the $\log$ function as:
\begin{align}
&\frac{1+\gamma_B+\sum_{M_i \in {\cal S}} \gamma_i }{(1+\Gamma_M)^S}=\frac{1+\sum_{M_i \in {\cal S}} \gamma_i }{(1+\Gamma_M)^S} + \frac{\gamma_B}{(1+\Gamma_M)^S} \nonumber \\
&  \stackrel{\textrm{(a)}}{\geq} 1 + \frac{\gamma_B}{(1+\Gamma_M)^S} \stackrel{\textrm{(b)}}{\geq} 1 + \frac{\gamma_B}{(1+\Gamma_M)^{k+1}}
\end{align}
where (a) follows from the assumption that (\ref{eq:ThmMACinequalityforS}) is satisfied for ${\cal S}$ and (b) follows from $S \leq k+1$. 

We have thus shown that the downlink rate $C \left ( \frac{\gamma_B}{(1+\Gamma_M)^{k+1}} \right )$ is always decodable at $U$. It only remains to show that this is the maximal possible rate with such a property. Let us choose 
\begin{equation} \label{eq:SelectedSNRsMinRate}
	\gamma_1=\gamma_2=\cdots =\gamma_{k+1}=\frac{(1+\Gamma_M)^{k+1}-1}{k+1}
\end{equation}
We first show that with such a choice of SNRs all $k+1$ signals of the MTDs can be decoded. Let ${\cal S}$ be any subset of $S$ devices. Then the inequality that needs to be satisfied for this set can written as:
\begin{align}
	1+S\frac{(1+\Gamma_M)^{k+1}-1}{k+1} & \geq (1+\Gamma_M)^S
\end{align} 
Consider the function $f(x)=\frac{(1+\Gamma_M)^{x}-1}{x}$. It can be shown that this function is monotonically increasing when $\Gamma_M>0$ and $x \geq 1$. We can then write the following inequality:
\begin{align}
	1+Sf(k+1) \stackrel{\textrm{(a)}}{\geq} 1+Sf(S)=1+S\frac{(1+\Gamma_M)^{S}-1}{S}= (1+\Gamma_M)^S
\end{align}
where (a) follows from $S \leq k+1$. Hence, all the MTDs are decodable. The bound imposed on the rate $R_B$ by the set ${\cal S}$ with cardinality $|{\cal S}|=S$ can be written as follows:
\begin{equation} \label{eq:AchievableRBbound}
	R_B \leq \log_2 \left ( \frac{1+\gamma_B+S \cdot \frac{(1+\Gamma_M)^{k+1}-1}{k+1}}{(1+\Gamma_M)^{S}} \right)
\end{equation}
It can be shown that the right-hand side of (\ref{eq:AchievableRBbound}) decreases as $S$ increases, such that it reaches its minimal value when $S=k+1$, which is 
\begin{equation}
	R_B \leq \log_2 \left ( 1+ \frac{\gamma_B}{(1+\Gamma_M)^{k+1}} \right)
\end{equation}
which proves the achievable rate $R_B$ for the selected SNRs in (\ref{eq:SelectedSNRsMinRate}).
\end{proof}

\section*{Acknowledgment}
The research presented in this paper was partly supported by the Danish Council for Independent Research (Det Frie Forskningsr\aa d) within the Sapere Aude Research Leader program, Grant No. 11-105159 ``Dependable Wireless Bits for Machine-to-Machine (M2M) Communications'' and performed partly in the framework of the FP7 project ICT-317669 METIS. The authors would like to acknowledge the contributions of their colleagues in METIS, although the views expressed are those of the authors and do not necessarily represent the project.

\end{document}